\documentclass[11pt]{article}
\usepackage{amsmath,amssymb,amsthm}
\usepackage{graphicx}
\usepackage{changes}
\usepackage{tikz}
\usepackage{pgfplots}
\usepackage{float} 

\pgfplotsset{compat=1.18} 
\newtheorem{proposition}{Proposition}
\theoremstyle{definition}

\title{ }

\author{}
\date{}

\begin{document}

\author{
Nazaria Solferino\\
Department of Economics, Statistics and Business\\
Universitas Mercatorum, Rome}

\title{An analytical model of Disequilibrium and decentralized productive Exploration}

\date{}
\maketitle

\begin{abstract}

This paper studies the economic role of persistent dispersion in allocations across agents. We develop a tractable model in which firms allocate resources under imperfect information and behavioral updating, generating sustained heterogeneity in beliefs and actions. While dispersion induces static misallocation, it also fosters decentralized experimentation, allowing the economy to explore a broader set of productive opportunities. We show that the economy converges to a stationary equilibrium with strictly positive dispersion and that, under plausible conditions, such disequilibrium can dominate the perfectly coordinated benchmark. The model provides a novel interpretation of observed dispersion in productivity and returns as reflecting both inefficiency and productive exploration. It also yields testable predictions linking dispersion to growth and innovation dynamics.


\end{abstract}

\section{Introduction}

A central question in macroeconomics concerns the role of dispersion in economic activity. A large empirical literature documents substantial heterogeneity in productivity, investment, and marginal returns across firms and sectors, even within narrowly defined industries. This dispersion is typically interpreted as evidence of misallocation: resources are not deployed where they yield the highest returns, leading to aggregate productivity losses. Seminal contributions such as \cite{restuccia2008policy}, \cite{hsieh2009misallocation}, and \cite{midrigan2014finance} quantify large efficiency gains from eliminating such distortions, establishing misallocation as a key determinant of aggregate productivity and cross-country income differences.

At the same time, a growing body of evidence challenges the view that dispersion is purely inefficient. Micro-level studies document that high dispersion often coexists with innovation, entry, and reallocation dynamics \cite{foster2008reallocation, syverson2011productivity, haltiwanger2016firm}. Moreover, recent theoretical work shows that dispersion may arise endogenously even in efficient environments. For example, \cite{kehrig2021good} demonstrate that dispersion in marginal revenue products can increase output when firms optimally adjust capital under frictions, while \cite{buera2011finance} and \cite{moll2014productivity} highlight how heterogeneous investment responses generate persistent dispersion with non-trivial aggregate implications.

These findings point to a fundamental tension. Standard equilibrium models associate efficiency with perfect coordination and equalization of marginal returns.  Real economies exhibit persistent heterogeneity that may reflect ongoing experimentation, learning, and structural change. This perspective is closely related to classical insights by \cite{hayek1945use}, who emphasized the role of dispersed information, and \cite{schumpeter1934theory}, who highlighted the importance of experimentation and innovation. More recent work on information frictions and expectations formation further stresses how heterogeneous beliefs shape aggregate outcomes \cite{angeletos2010noisy, maćkowiak2009optimal, bordalo2018diagnostic, coibion2015information}.

Despite these advances, existing frameworks face an important limitation. Most models either interpret dispersion as a deviation from an efficient benchmark driven by distortions, or generate dispersion as an equilibrium outcome without explicitly modeling the trade-off between coordination and decentralized exploration. In particular, they struggle to capture the idea that disequilibrium may itself be a productive feature of the economy.

This paper proposes a tractable theoretical framework that places this trade-off at the center of the analysis.
We interpret the economy as populated by a large number of agents-firms, investors, or managers* who must allocate resources across a continuum of productive opportunities (e.g., technologies, sectors, or strategies). In each period, there exists an underlying fundamental that determines where returns are highest. However, this fundamental is not directly observable. Instead, agents rely on imperfect signals and form beliefs using updating rules that may deviate from full rationality.

As a result, agents make heterogeneous decisions: some overinvest in unproductive directions, while others underinvest or experiment with alternative opportunities. From a static perspective, this dispersion in actions generates misallocation and reduces efficiency. In standard models, equilibrium is characterized by full coordination: agents share common beliefs, dispersion vanishes, and resources are allocated efficiently.

The key departure in our framework is that dispersion is not purely wasteful. When agents choose different actions, the economy effectively conducts decentralized experimentation. Some agents may discover high-return opportunities that would not be explored under perfect coordination. In this sense, what appears as misallocation ex post may reflect valuable exploration ex ante.

Our model formalizes this trade-off between coordination and exploration. On the one hand, dispersion in actions leads to inefficient allocation relative to the current fundamental. On the other hand, dispersion generates informational and productive benefits that are absent in a perfectly coordinated equilibrium. By introducing persistent heterogeneity in beliefs through imperfect information and behavioral noise, the model delivers a stationary environment in which disequilibrium, understood as cross-sectional dispersion in actions, is not a transient deviation, but a structural feature of the economy.

This framework allows us to characterize the conditions under which equilibrium allocations are dominated by outcomes with persistent dispersion, thereby providing a micro-founded rationale for why economies may benefit from sustained deviations from perfect coordination.

Our main contribution is to provide a novel interpretation of misallocation measures. In our framework, observed dispersion reflects a combination of inefficient allocation and productive experimentation. This perspective helps reconcile empirical findings that link high dispersion to both low contemporaneous efficiency and high future growth, and connects to recent work on heterogeneous expectations and belief-driven fluctuations \cite{bordalo2018diagnostic, coibion2015information}.

More broadly, the paper contributes to several strands of the literature. First, it relates to the misallocation literature \cite{restuccia2008policy, hsieh2009misallocation, midrigan2014finance}. Second, it connects to the literature on heterogeneous firms and produivity dynamics \cite{foster2008reallocation, syverson2011productivity}. Third, it contributes to models of information frictions and expectations formation \cite{angeletos2010noisy, maćkowiak2009optimal}. Finally, it speaks to a growing literature on the role of heterogeneity and non-coordination in shaping aggregate outcomes.

The remainder of the paper is organized as follows. Section 2 introduces the baseline model and characterizes the stationary distribution of beliefs and allocations. Section 3 provides a  welfare comparison with  the state of equilibrium. Section 4 discusses policy implications. Section 5  concludes.

\section{Baseline Model: Disequilibrium, Misallocation, and Productive Dispersion}

This section develops a tractable framework to study how dispersed, imperfectly coordinated allocations of resources can arise endogenously and, under certain conditions, improve aggregate outcomes relative to standard equilibrium benchmarks.

For simplicity,we assume that time is discrete, $t = 0,1,2,\dots$. The economy is populated by a continuum of agents $i \in [0,1]$. Each period, a latent fundamental $\theta_t \in \mathbb{R}$ determines the location of productive opportunities (e.g., the direction of technological change or sectoral profitability).

The fundamental evolves according to the following process:
\begin{equation}
\theta_{t+1} = \rho \theta_t + \varepsilon_{t+1}, \quad |\rho| < 1, \quad \varepsilon_t \sim \mathcal{N}(0,\sigma_\varepsilon^2).
\end{equation}

Agents must choose an action $a_{i,t} \in \mathbb{R}$, which we interpret as the allocation of capital or effort across a one-dimensional space of technologies or sectors.

Agent $i$'s payoff is given by:
\begin{equation}
\pi_{i,t} = - (a_{i,t} - \theta_t)^2.
\end{equation}

This quadratic loss captures misallocation: output is maximized when agents allocate exactly at the fundamental $\theta_t$. Deviations represent inefficient allocation of resources.

Agents do not observe $\theta_t$ directly. Instead, they receive a noisy private signal:
\begin{equation}
s_{i,t} = \theta_t + \nu_{i,t}, \quad \nu_{i,t} \sim \mathcal{N}(0,\sigma_\nu^2),
\end{equation}
with $\nu_{i,t}$ i.i.d. across agents and time.

Agents also form beliefs $\hat{\theta}_{i,t}$ about the fundamental using a behavioral updating rule:
\begin{equation}
\hat{\theta}_{i,t+1} = (1-\alpha)\hat{\theta}_{i,t} + \alpha s_{i,t} + \eta_{i,t+1},
\end{equation}
where $\alpha \in (0,2)$ governs the speed of adjustment and $\eta_{i,t} \sim \mathcal{N}(0,\sigma_\eta^2)$ represents non-vanishing cognitive or behavioral noise.

This specification captures the fact that even in the limit of abundant information, idiosyncratic errors remain.

Given quadratic payoffs, optimal actions are:
\begin{equation}
a_{i,t} = \hat{\theta}_{i,t}.
\end{equation}

Thus, heterogeneity in beliefs translates directly into heterogeneity in allocations.

Define the cross-sectional mean and variance of beliefs:
\begin{equation}
m_t = \mathbb{E}[\hat{\theta}_{i,t}], \quad v_t = \mathrm{Var}(\hat{\theta}_{i,t}).
\end{equation}

The law of motion for the mean is:
\begin{equation}
m_{t+1} = (1-\alpha)m_t + \alpha \theta_t.
\end{equation}

The variance evolves according to:
\begin{equation}
v_{t+1} = (1-\alpha)^2 v_t + \alpha^2 \sigma_\nu^2 + \sigma_\eta^2.
\end{equation}

\begin{proposition}[Ergodicity]
Suppose $|\rho|<1$ and $\alpha \in (0,2)$. Then the stochastic process $(\theta_t, m_t, v_t)$ admits a unique invariant distribution. Moreover, $v_t$ converges deterministically to a unique fixed point $v^*$.
\end{proposition}

\begin{proof}
The process $\theta_t$ is a stationary Gaussian AR(1). The mean $m_t$ evolves as a stable linear function of $(m_{t-1}, \theta_t)$, implying stationarity. The variance follows a deterministic linear recursion with coefficient $(1-\alpha)^2<1$, hence converges to a unique fixed point. Standard results for affine stochastic systems imply existence and uniqueness of an invariant distribution.
\end{proof}

The steady-state variance is given by:
\begin{equation}
v^* = \frac{\alpha^2 \sigma_\nu^2 + \sigma_\eta^2}{2\alpha - \alpha^2}.
\end{equation}

Importantly, even if $\sigma_\nu \to 0$, dispersion remains strictly positive whenever $\sigma_\eta > 0$.

Aggregate output is defined as:
\begin{equation}
Y_t = - \int_0^1 (a_{i,t} - \theta_t)^2 di + \gamma \Omega(v_t),
\end{equation}
where $\gamma > 0$ and $\Omega(\cdot)$ is increasing and twice differentiable.

Using the law of total variance:
\begin{equation}
\int (a_{i,t} - \theta_t)^2 di = (\bar{a}_t - \theta_t)^2 + v_t,
\end{equation}
so that:
\begin{equation}
Y_t = - (\bar{a}_t - \theta_t)^2 - v_t + \gamma \Omega(v_t).
\end{equation}

The first term captures aggregate misallocation, the second dispersion costs, and the third the benefits of decentralized experimentation.

Consider the benchmark case $\sigma_\eta = 0$. Then:
\begin{equation}
v^*_{eq} = \frac{\alpha^2 \sigma_\nu^2}{2\alpha - \alpha^2}.
\end{equation}

In the limit of perfect information ($\sigma_\nu \to 0$), dispersion vanishes:
\begin{equation}
v^*_{eq} \to 0.
\end{equation}

This corresponds to a standard equilibrium with perfect coordination and no cross-sectional heterogeneity.

\section{Disequilibrium and Welfare Comparison}

Define expected welfare:
\begin{equation}
W(v) = - v + \gamma \Omega(v).
\end{equation}

\begin{proposition}[Productive Disequilibrium]
Suppose $\Omega(\cdot)$ is increasing and there exists $v>0$ such that $\gamma \Omega(v) > v$. Then there exists $\sigma_\eta > 0$ such that:
\begin{equation}
W(v^*) > W(v^*_{eq}),
\end{equation}
i.e., an economy with persistent dispersion yields higher welfare than the equilibrium benchmark.
\end{proposition}

\begin{proof}
Since $v^* > v^*_{eq}$ for any $\sigma_\eta > 0$, the welfare difference is:
\begin{equation}
W(v^*) - W(v^*_{eq}) = - (v^* - v^*_{eq}) + \gamma [\Omega(v^*) - \Omega(v^*_{eq})].
\end{equation}
By continuity of $\Omega$, if $\gamma \Omega(v)$ dominates $v$ locally, the expression is positive for some $v^*>v^*_{eq}$.
\end{proof}

\begin{proposition}[Optimal Disequilibrium]
If $\Omega$ is strictly concave for large $v$ and satisfies $\Omega'(0)>0$, then there exists a unique $v^{opt} > 0$ maximizing $W(v)$. This corresponds to an interior level of behavioral noise $\sigma_\eta^* > 0$.
\end{proposition}

\begin{proof}
$W'(v) = -1 + \gamma \Omega'(v)$. By assumptions, $W'(0) > 0$ and $W'(v) \to -1$ as $v \to \infty$. By continuity, there exists a unique interior maximizer.
\end{proof}

The model implies that equilibrium allocations with minimal dispersion may be inefficient when dispersion generates productive experimentation. Behavioral noise sustains heterogeneity in actions, which can improve aggregate outcomes despite static inefficiencies.

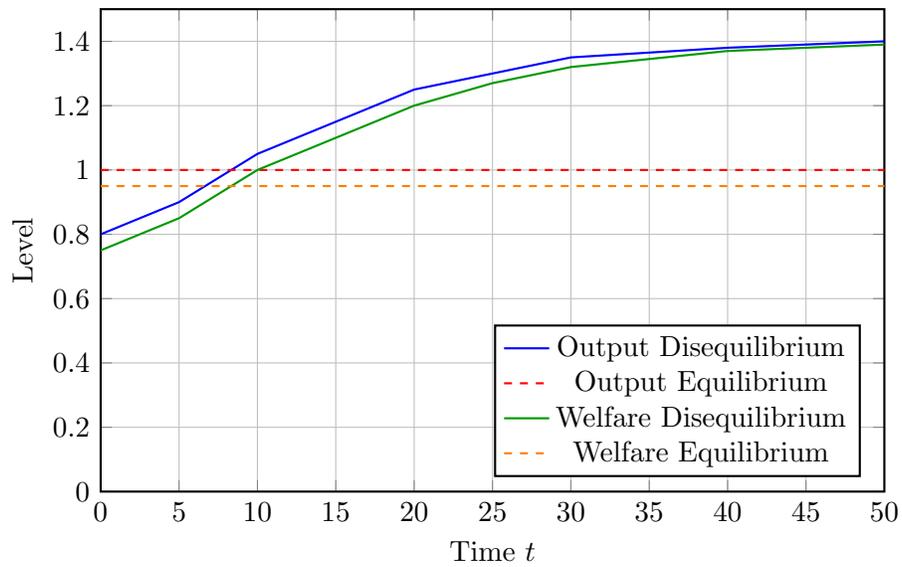
\begin{figure}[H]
\centering
\begin{tikzpicture}
\begin{axis}[
    width=12cm,
    height=8cm,
    xlabel={Time $t$},
    ylabel={Level},
    xmin=0, xmax=50,
    ymin=0, ymax=1.5,
    legend pos=south east,
    grid=both,
    thick
]

\addplot[blue, thick] coordinates {
(0,0.8)(5,0.9)(10,1.05)(15,1.15)(20,1.25)(25,1.3)(30,1.35)(40,1.38)(50,1.4)
};
\addlegendentry{Output Disequilibrium}

\addplot[red, dashed, thick] coordinates {
(0,1)(5,1)(10,1)(15,1)(20,1)(25,1)(30,1)(40,1)(50,1)
};
\addlegendentry{Output Equilibrium}

\addplot[green!60!black, thick] coordinates {
(0,0.75)(5,0.85)(10,1.0)(15,1.1)(20,1.2)(25,1.27)(30,1.32)(40,1.37)(50,1.39)
};
\addlegendentry{Welfare Disequilibrium}

\addplot[orange, dashed, thick] coordinates {
(0,0.95)(5,0.95)(10,0.95)(15,0.95)(20,0.95)(25,0.95)(30,0.95)(40,0.95)(50,0.95)
};
\addlegendentry{Welfare Equilibrium}

\end{axis}
\end{tikzpicture}
\caption{Output and welfare dynamics under equilibrium and disequilibrium. Disequilibrium initially entails lower output and welfare but surpasses equilibrium in the long run due to decentralized experimentation.}
\end{figure}

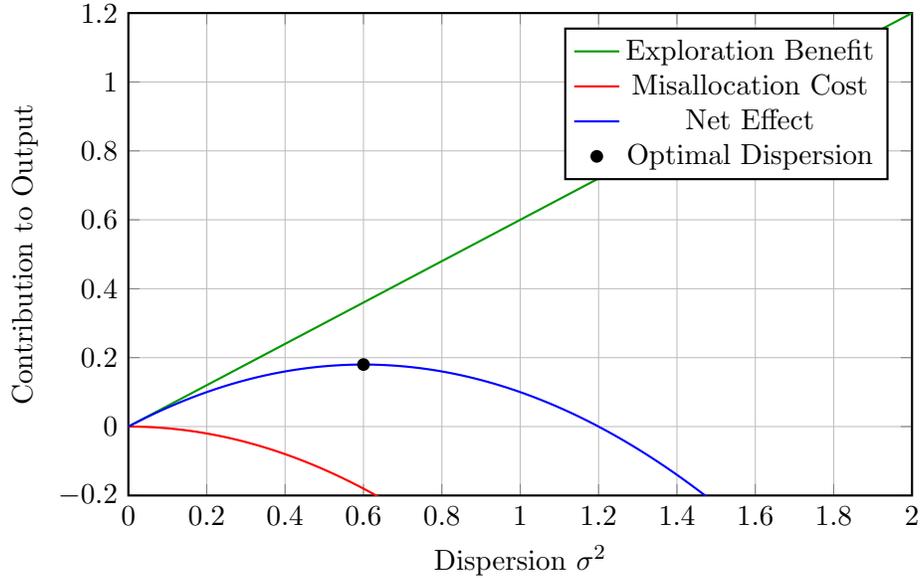
\begin{figure}[H]
\centering
\begin{tikzpicture}
\begin{axis}[
    width=12cm,
    height=8cm,
    xlabel={Dispersion $\sigma^2$},
    ylabel={Contribution to Output},
    xmin=0, xmax=2,
    ymin=-0.2, ymax=1.2,
    domain=0:2,
    samples=100,
    thick,
    legend pos=north east,
    grid=both
]

\addplot[green!60!black, thick] {0.6*x};
\addlegendentry{Exploration Benefit}

\addplot[red, thick] {-0.5*x^2};
\addlegendentry{Misallocation Cost}

\addplot[blue, thick] {0.6*x - 0.5*x^2};
\addlegendentry{Net Effect}

\addplot[only marks, mark=*, mark size=2pt, black] coordinates {(0.6,0.18)};
\addlegendentry{Optimal Dispersion}

\end{axis}
\end{tikzpicture}
\caption{Trade-off between exploration and misallocation. Aggregate output is maximized at an interior level of dispersion, illustrating that disequilibria can be beneficial.}
\end{figure}

\section{Discussion and Policy Implications}

The analysis developed in this work has several implications for how economists interpret dispersion in economic activity and, more broadly, for the design of policy interventions. A central message is that dispersion in allocations—traditionally viewed as a symptom of inefficiency—may instead reflect a fundamental trade-off between coordination and experimentation. This perspective calls for a reassessment of policies aimed at reducing heterogeneity in economic outcomes.

In standard frameworks, policy interventions that improve information, reduce frictions, or enhance coordination are unambiguously welfare improving. By contrast, in our setting, such interventions may have non-monotonic effects. Reducing informational frictions or behavioral noise compresses the dispersion of beliefs and actions, thereby improving static allocative efficiency. However, it simultaneously reduces the extent of decentralized experimentation. As a result, policies that push the economy closer to perfect coordination may inadvertently eliminate mechanisms through which new productive opportunities are discovered.

This trade-off is particularly relevant in environments characterized by rapid technological change or high uncertainty about future growth opportunities. In such contexts, a certain degree of heterogeneity in decisions may be desirable, as it ensures that different parts of the economy explore alternative directions. From this perspective, dispersion can be interpreted as a form of endogenous diversification, akin to portfolio allocation at the aggregate level. Policies that excessively standardize behavior—through rigid regulations, uniform incentives, or overly precise guidance—may therefore reduce the economy’s ability to adapt and innovate.

The model also has implications for the interpretation of misallocation measures. Empirical indicators based on dispersion in marginal products or returns are often used to quantify inefficiencies and to guide policy recommendations. Our framework suggests that such measures may conflate two distinct components: inefficient allocation and productive experimentation. As a result, a reduction in measured dispersion is not necessarily welfare improving, and policies targeting dispersion should be evaluated in light of their effects on both dimensions. This insight is particularly relevant for development and industrial policy, where efforts to reallocate resources toward high-productivity firms may come at the cost of reducing experimentation and entry.

Another implication concerns information policy. A large body of work emphasizes the benefits of improving the accuracy and dissemination of information, for example through transparency policies, disclosure requirements, or data provision. While such policies reduce uncertainty and improve coordination, our analysis highlights a potential downside: by aligning beliefs too closely, they may reduce the diversity of actions in the economy. This suggests that the optimal design of information policies may involve a trade-off between precision and diversity, rather than a simple objective of maximizing informational accuracy.

The framework also speaks to the role of behavioral biases and non-rational expectations. While these are typically viewed as distortions that should be corrected, our results indicate that some forms of behavioral noise may have beneficial aggregate effects by sustaining heterogeneity. This does not imply that all biases are desirable, but rather that their aggregate consequences depend on the environment. In particular, when the benefits of experimentation are large, eliminating behavioral heterogeneity may reduce welfare. This perspective provides a more nuanced view of behavioral interventions and suggests that policies aimed at “debiasing” agents should take into account their impact on aggregate exploration.

Finally, the model has implications for macroeconomic stabilization policy. In environments with high dispersion, aggregate outcomes may be more resilient to shocks, as different agents respond in diverse ways. Conversely, highly coordinated economies may be more fragile, as they concentrate resources in a narrow set of activities. This suggests that some degree of heterogeneity may enhance robustness, providing a form of insurance against model misspecification or unforeseen shocks.

Overall, the analysis highlights that the welfare effects of dispersion depend critically on the balance between its costs and benefits. Policies that focus exclusively on reducing misallocation may overlook the role of dispersion as a driver of experimentation and long-run growth. A more comprehensive approach should recognize this dual role and aim to design interventions that preserve beneficial heterogeneity while mitigating its most harmful aspects.

\section{Conclusion}

This paper develops a theoretical framework to study the role of persistent dispersion in economic activity. Departing from standard equilibrium models, we introduce a setting in which agents form beliefs under imperfect information and behavioral updating, generating sustained heterogeneity in actions. This heterogeneity gives rise to a fundamental trade-off: while dispersion leads to misallocation relative to the current fundamental, it also enables decentralized experimentation and the exploration of new productive opportunities.

We show that the economy converges to a stationary equilibrium with strictly positive dispersion and that, under plausible conditions, such disequilibrium can dominate the perfectly coordinated benchmark. In this sense, the model provides a formal rationale for why economies may benefit from sustained deviations from full coordination. Rather than being purely a source of inefficiency, dispersion emerges as an integral component of the economic system, supporting adaptation, innovation, and long-run growth.

The framework also offers a new perspective on empirical measures of misallocation, suggesting that observed dispersion reflects a combination of inefficiency and productive exploration. This interpretation helps reconcile evidence linking dispersion to both low contemporaneous productivity and high future growth, and highlights the importance of distinguishing between different sources of heterogeneity.

More broadly, the paper contributes to a growing literature that reconsiders the role of heterogeneity and disequilibrium in economic systems. By emphasizing the productive aspects of dispersion, it challenges the traditional view that efficiency requires full coordination, and instead highlights the potential benefits of maintaining a diverse and decentralized economic structure. 
The framework developed in this paper opens several avenues for future research. Studying the interaction between policy interventions and endogenous disequilibrium could yield practical guidance for designing policies that balance coordination with the benefits of decentralized exploration.

\end{document}